\documentclass{acmtog}

\acmVolume{VV}
\acmNumber{N}
\acmYear{YYYY}
\acmMonth{Month}
\acmArticleNum{XXX}  
\acmdoi{10.1145/XXXXXXX.YYYYYYY}

\usepackage{amsmath}
\usepackage{amssymb}
\usepackage{algorithm}
\usepackage{algorithmic}
\usepackage{enumerate}
\usepackage{picinpar}
\usepackage{color}
\usepackage{hyperref}
\usepackage{graphicx}

\hypersetup{
   colorlinks=false,
   pdfborder={0 0 0}
}

\definecolor{orange}{rgb}{.8,.4,0}

\def\RR{\mathbb{R}}
\def\id{\mathrm{id}}
\def\etal{\emph{et al.}}
\def\ie{\emph{i.e.}}
\def\eg{\emph{e.g.}}
\def\etc{\emph{etc.}}
\newcommand{\appn}[1]{Appendix~\ref{app:#1}}
\newcommand{\eqn}[1]{Eq.~(\ref{eq:#1})}
\newcommand{\figr}[1]{Figure~\ref{fig:#1}}
\newcommand{\figrs}[2]{Figures~\ref{fig:#1} and \ref{fig:#2}}
\newcommand{\sect}[1]{Section~\ref{sec:#1}}
\newcommand{\sects}[2]{Sections~\ref{sec:#1} and \ref{sec:#2}}
\newcommand{\algo}[1]{Algorithm~\ref{alg:#1}}
\newcommand{\tabl}[1]{Table~\ref{tab:#1}}
\newcommand{\lemm}[1]{Lemma~\ref{lem:#1}}

\begin{document}

\markboth{K. Crane et al.}{Geodesics in Heat}

\title{Geodesics in Heat}

\author{KEENAN CRANE
\affil{Caltech}
CLARISSE WEISCHEDEL, MAX WARDETZKY
\affil{University of G\"{o}ttingen}}



\maketitle

\begin{bottomstuff} 
\end{bottomstuff}

\begin{abstract} 
We introduce the \emph{heat method} for computing the shortest geodesic distance to a specified subset (\eg, point or curve) of a given domain.  The heat method is robust, efficient, and simple to implement since it is based on solving a pair of standard linear elliptic problems.  The method represents a significant breakthrough in the practical computation of distance on a wide variety of geometric domains, since the resulting linear systems can be prefactored once and subsequently solved in near-linear time.  In practice, distance can be updated via the heat method an order of magnitude faster than with state-of-the-art methods while maintaining a comparable level of accuracy.  We provide numerical evidence that the method converges to the exact geodesic distance in the limit of refinement; we also explore smoothed approximations of distance suitable for applications where more regularity is required.
\end{abstract}

\section{Introduction}
\label{sec:intro}

Imagine touching a scorching hot needle to a single point on a surface.  Over
time heat spreads out over the rest of the domain and can be described by a
function \(k_{t,x}(y)\) called the \emph{heat kernel}, which measures the heat
transferred from a source \(x\) to a destination \(y\) after time \(t\).  A well-known
relationship between heat and distance is Varadhan's
formula~\shortcite{Varadhan:1967:OTB}, which says that the geodesic distance \(\phi\)
between any pair of points \(x,y\) on a Riemannian manifold can be recovered via
a simple pointwise transformation of the heat kernel: \begin{equation} \phi(x,y) = \lim_{t \rightarrow 0} \sqrt{-4t
\log k_{t,x}(y)}. \label{eq:varadhan}\end{equation} The intuition behind this behavior
stems from the fact that heat diffusion can be modeled as a large collection of
hot particles taking random walks starting at \(x\): any particle that
reaches a distant point \(y\) after a small time \(t\) has had little time to deviate
from the shortest possible path.  To date, however, this relationship has not been
exploited by numerical algorithms that compute geodesic distance.

\begin{figure}
\begin{center}
\includegraphics[width=.75\columnwidth]{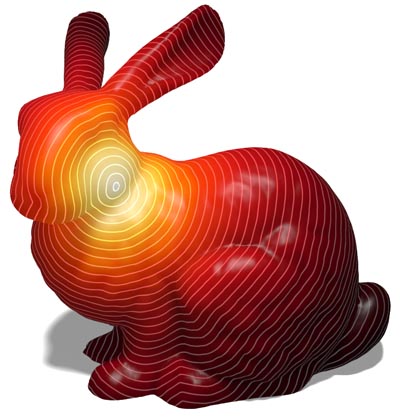}
\end{center}
\caption{Geodesic distance on the Stanford Bunny.  The heat method allows distance to be rapidly updated for new source points or curves.}
\label{fig:bunny}
\end{figure}

Why has Varadhan's formula been overlooked in this context?  The main reason,
perhaps, is that it requires a precise numerical reconstruction of the heat
kernel, which is difficult to obtain for small values of \(t\) --
applying the formula to a mere approximation of \(k_{t,x}\) does not yield
the correct result, as illustrated in \figrs{motivation}{logvsheat}.  The
main idea behind the heat method is to circumvent this issue by working
with a broader class of inputs, namely any function whose gradient is parallel
to geodesics.  We can then separate the computation of distance into
two separate stages: first compute the gradient of the distance field, then
recover the distance itself.

\begin{figure}
\begin{center}
\includegraphics{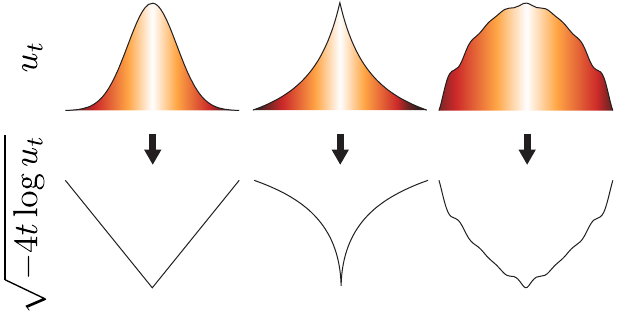}
\end{center}
\caption{Given an \emph{exact} reconstruction of the heat kernel \emph{(top left)} Varadhan's formula can be used to recover geodesic distance \emph{(bottom left)} but fails in the presence of approximation or numerical error \emph{(middle, right)}, as shown here for a point source in 1D.  The robustness of the heat method stems from the fact that it depends only on the \emph{direction} of the gradient.}
\label{fig:motivation}
\end{figure}

Relative to existing algorithms, the heat method offers two major advantages.
First, it can be applied to virtually any type of geometric discretization,
including regular and irregular grids, polygonal meshes, and even
unstructured point clouds. Second, it involves only the solution of sparse
linear systems, which can be prefactored once and rapidly re-solved many
times.  This feature makes the heat method particularly valuable for
applications such as shape matching, path planning, and level set-based
simulation (\eg, free-surface fluid flows), which require repeated distance
queries on a fixed geometric domain.  Moreover, because linear elliptic
equations are widespread in scientific computing, the heat method can
immediately take advantage of new developments in numerical linear algebra
and parallelization.

\section{Related Work}
\label{sec:related}

The prevailing approach to distance computation is to solve the \emph{eikonal equation}
\begin{equation}
\label{eq:eikonal}
|\nabla \phi| = 1
\end{equation}
subject to boundary conditions \(\phi|_{\gamma} = 0\) over some subset \(\gamma\) of the domain.  This formulation is \emph{nonlinear} and \emph{hyperbolic}, making it difficult to solve directly.  Typically one applies an iterative relaxation scheme such as Gauss-Seidel -- special update orders are known as \emph{fast marching} and \emph{fast sweeping}, which are some of the most popular algorithms for distance computation on regular grids~\cite{Sethian96} and triangulated surfaces~\cite{KimmelSethian98}.  These algorithms can also be used on implicit surfaces~\cite{MemoliSapiro01}, point clouds~\cite{MemoliSapiro05}, and polygon soup~\cite{Campen:2011:WOB}, but only indirectly: distance is computed on a simplicial mesh or regular grid that approximates the original domain.  Implementation of fast marching on simplicial grids is challenging due to the need for nonobtuse triangulations (which are notoriously difficult to obtain) or else a complex unfolding procedure to preserve monotonicity of the solution; moreover these issues are not well-studied in dimensions greater than two.  Fast marching and fast sweeping have asymptotic complexity of \(O(n \log n)\) and \(O(n)\), respectively, but sweeping is often slower due to the large number of sweeps required to obtain accurate results~\cite{Hysing:2005:TEE}.

The main drawback of these methods is that they do not reuse information: the distance to different subsets \(\gamma\) must be computed entirely from scratch each time.  Also note that both sweeping and marching present challenges for parallelization: priority queues are inherently serial, and irregular meshes lack a natural sweeping order.  Weber \etal~\shortcite{Weber08} address this issue by decomposing surfaces into regular grids, but this decomposition resamples the surface and requires a low-distortion parameterization that may be difficult to obtain (note that the heat method would also benefit from such a decomposition).

In a different development, Mitchell \etal~\shortcite{Mitchell87} give an \(O(n^2 \log n)\) algorithm for computing the exact polyhedral distance from a single source to all other vertices of a triangulated surface.  Surazhsky \etal~\shortcite{Surazhsky05} demonstrate that this algorithm tends to run in sub-quadratic time in practice, and present an approximate \(O(n \log n)\) version of the algorithm with guaranteed error bounds; Bommes and Kobbelt~\shortcite{BommesKobbelt07} extend the algorithm to polygonal sources.  Similar to fast marching, these algorithms propagate distance information in wavefront order using a priority queue, again making them difficult to parallelize.  More importantly, the amortized cost of these algorithms (over many different source subsets \(\gamma\)) is substantially greater than for the heat method since they do not reuse information from one subset to the next.  Finally, although \cite{Surazhsky05} greatly simplifies the original formulation, these algorithms remain challenging to implement and do not immediately generalize to domains other than triangle meshes.

\begin{figure}[t]
\begin{center}
\includegraphics[width=.52\columnwidth]{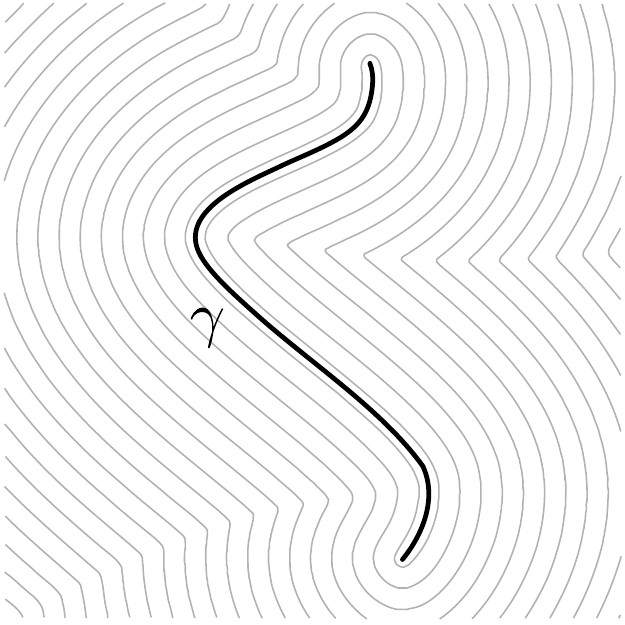}
\end{center}
\caption{The heat method computes the shortest distance to a subset \(\gamma\) of a given domain.  Gray curves indicate isolines of the distance function.}
\label{fig:curve}
\end{figure}

Closest to our approach is the recent method of Rangarajan and Gurumoorthy~\shortcite{Rangarajan:2011:AFE}, who do not appear to be aware of Varadahn's formula -- they instead derive an analogous relationship \(\phi = -\sqrt{\hbar\log \psi}\) between the distance function and solutions \(\psi\) to the time-independent Schr\"{o}dinger equation.  We emphasize, however, that this derivation applies only in \(\RR^n\) where \(\psi\) takes a special form -- in this case it may be just as easy to analytically invert the Euclidean heat kernel \(u_{t,x} = (4 \pi t)^{-n/2} e^{-\phi(x,y)^2/4t}\).  Moreover, they compute solutions using the fast Fourier transform, which limits computation to regular grids.  To obtain accurate results their method requires either the use of arbitrary-precision arithmetic or a combination of multiple solutions for various values of \(\hbar\); no general guidance is provided for determining appropriate values of \(\hbar\).

Finally, there is a large literature on \emph{smooth distances}~\cite{Coifman06,Fouss:2007:RWC,Lipman:2010:BD}, which are valuable in contexts where differentiability is required.  However, existing smooth distances may not be appropriate in contexts where the \emph{geometry} of the original domain is important, since they do not attempt to approximate the original metric and therefore substantially violate the unit-speed nature of geodesics (\figr{biharmonic}).  Interestingly enough, these distances also have an interpretation in terms of simple discretizations of heat flow -- see \sect{smoothed} for further discussion.

\begin{figure}[b]
\begin{center}
\includegraphics[width=\columnwidth]{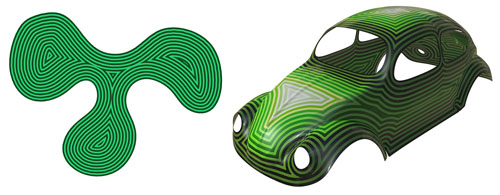}
\end{center}
\caption{Distance to the boundary on a region in the plane (left) or a surface in \(\mathbb{R}^3\) is achieved by simply placing heat along the boundary curve.  Note good recovery of the \emph{cut locus}, \ie, points with more than one closest point on the boundary.\label{fig:boundary}}
\end{figure}

\section{The Heat Method}
\label{sec:algorithm}

Our method can be described purely in terms of operations on smooth manifolds; we explore spatial and temporal discretization in \sects{temporal}{spatial}, respectively. Let \(\Delta\) be the negative-semidefinite Laplace--Beltrami operator acting on  (weakly) differentiable real-valued functions over a Riemannian manifold \((M,g)\).  The heat method consists of three basic steps:

\begin{algorithm}
\caption{The Heat Method}
\label{eq:fpm}
\begin{enumerate}[I.]
\item Integrate the heat flow \(\dot{u} = \Delta u\) for some fixed time \(t\).
\item Evaluate the vector field \(X = -\nabla u / |\nabla u|\).
\item Solve the Poisson equation \(\Delta \phi = \nabla \cdot X\).
\end{enumerate}
\label{alg:heatmethod}
\end{algorithm}

\begin{figure}[t]
\begin{center}
\includegraphics[width=\columnwidth]{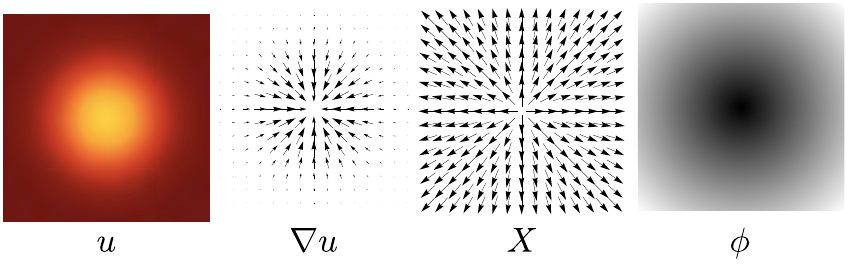}
\end{center}
\caption{Outline of the heat method. (I) Heat \(u\) is allowed to diffuse for a brief period of time \emph{(left)}. (II) The temperature gradient \(\nabla u\) \emph{(center left)} is normalized and negated to get a unit vector field \(X\) \emph{(center right)} pointing along geodesics. (III) A function~\(\phi\) whose gradient follows \(X\) recovers the final distance \emph{(right)}.}
\label{fig:algorithm}
\end{figure}

The function \(\phi\) approximates geodesic distance, approaching the true distance as \(t\) goes to zero (\eqn{varadhan}).  Note that the solution to step III is unique only up to an additive constant -- final values simply need to be shifted such that the smallest distance is zero.  Initial conditions \(u_0 = \delta(x)\) (\ie, a Dirac delta) recover the distance to a single source point \(x \in M\) as in \figr{bunny}, but in general we can compute the distance to any piecewise submanifold \(\gamma\) by setting \(u_0\) to a generalized Dirac~\cite{Villa:2006:MGM} over \(\gamma\) (see \figrs{curve}{boundary}).

The heat method can be motivated as follows.  Consider an approximation \(u_t\) of heat flow for a fixed time \(t\).  Unless \(u_t\) exhibits \emph{precisely} the right rate of decay, Varadhan's transformation \(u_t \mapsto \sqrt{-4t \log u_t}\) will yield a poor approximation of the true geodesic distance \(\phi\) because it is highly sensitive to errors in magnitude (see \figrs{motivation}{logvsheat}).  The heat method asks for something different: it asks only that the gradient \(\nabla u_t\) points in the right direction, \ie, parallel to \(\nabla\phi\).  Magnitude can safely be ignored since we know (from the eikonal equation) that the gradient of the true distance function has unit length.  We therefore compute the normalized gradient field \(X = -\nabla u/|\nabla u|\) and find the closest scalar potential \(\phi\) by minimizing \(\int_M|\nabla \phi - X|^2\), or equivalently, by solving the corresponding Euler-Lagrange equations \(\Delta \phi = \nabla \cdot X\)~\cite{Schwarz:1995:HDM}.  The overall procedure is depicted in \figr{algorithm}.

\begin{figure}
\begin{center}
\includegraphics[width=.75\columnwidth]{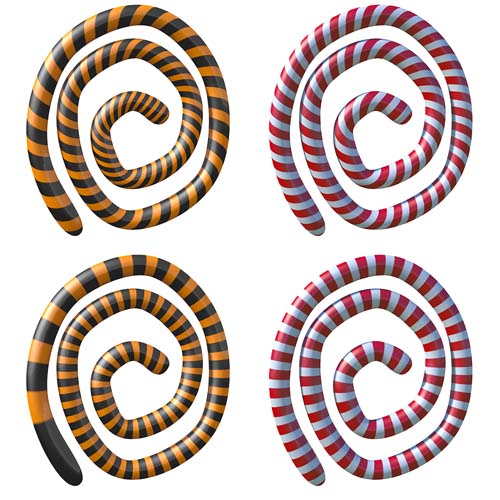}
\end{center}
\caption{\emph{Left:} Varadhan's formula. \emph{Right:} the heat method.  Even for very small values of \(t\), simply applying Varadhan's formula does not provide an accurate approximation of geodesic distance \emph{(top left)}; for large values of \(t\) spacing becomes even more uneven \emph{(bottom left)}.  Normalizing the gradient results in a more accurate solution, as indicated by evenly spaced isolines \emph{(top right)}, and is also valuable when constructing a smoothed distance function \emph{(bottom right)}.}
\label{fig:logvsheat}
\end{figure}

\subsection{Time Discretization}
\label{sec:temporal}

We discretize the heat equation from step I of \algo{heatmethod} in time using a single backward Euler step for some fixed time \(t\).  In practice, this means we simply solve the linear equation
\begin{align}\label{eq:IE-step-1}
(\id - t\Delta) u_t = u_0
\end{align}
over the entire domain \(M\), where \(\id\) is the identity (here we still consider a smooth manifold; spatial discretization is discussed in \sect{spatial}).  Note that backward Euler provides a maximum principle, preventing spurious oscillations in the solution~\cite{Wade05}.  We can get a better understanding of solutions to \eqn{IE-step-1} by considering the elliptic boundary value problem
\begin{align}\label{eq:BVP-IE}
\begin{split}
(\id-t\Delta)v_t &= 0 \quad\text{on}\quad M\backslash\gamma\\
v_t &=1 \quad\text{on}\quad \gamma \ .
\end{split}
\end{align}
which for a point source yields a solution \(v_t\) equal to \(u_t\) up to a multiplicative constant.  As established by Varadhan in his proof of \eqn{varadhan}, \(v_t\) also has a close relationship with distance, namely
\begin{align}\label{eq:conv-Varadhan-2}
\lim_{t\to 0} -\tfrac{\sqrt{t}}{2}\log v_t = \phi.
\end{align}
This relationship ensures the validity of steps II and III since the transformation applied to \(v_t\) preserves the direction of the gradient.

\subsection{Spatial Discretization}
\label{sec:spatial}

In principle the heat method can be applied to any domain with a discrete gradient (\(\nabla\)), divergence (\(\nabla \cdot\)) and Laplace operator (\(\Delta\)).  Here we investigate several possible discretizations.

\subsubsection{Simplicial Meshes}
\label{sec:simplicial}

Let \(u \in \mathbb{R}^{|V|}\) specify a piecewise linear function on a triangulated surface.  A standard discretization of the Laplacian at a vertex \(i\) is given by
\[ \mbox{\hspace{-.75in}}(L u)_i = \frac{1}{2A_i}\sum_j (\cot\alpha_{ij} + \cot\beta_{ij})(u_j-u_i), \]
\begin{window}[0,r,\includegraphics[trim = 0 0 0 0]{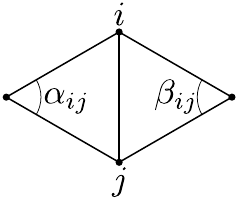},{}]
where \(A_i\) is one third the area of all triangles incident on vertex \(i\), the sum is taken over all neighboring vertices \(j\), and \(\alpha_{ij},\beta_{ij}\) are the angles opposing the corresponding edge~\cite{MacNeal:1949:SPD}.  We can express this operation via a matrix \(L = A^{-1} L_C\), where \(A \in \mathbb{R}^{|V|\times|V|}\) is a diagonal matrix containing the vertex areas and \(L_C \in \mathbb{R}^{|V|\times|V|}\) is the \emph{cotan operator} representing the remaining sum.  Heat flow can then be computed by solving the symmetric positive-definite system
\end{window}
\[ ( A - t L_C) u = u_0 \]
\begin{window}[0,r,\includegraphics[trim = 0 0 0 0]{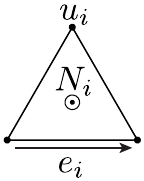},{}]
where \(u_0\) is a Kronecker delta over \(\gamma\) (\ie, one for source vertices; zero otherwise).  The gradient in a given triangle can be expressed succinctly as
\end{window}
\[ \nabla u = \frac{1}{2A_f} \sum_i u_i (N \times e_i) \]
\begin{window}[0,r,\includegraphics[trim = 0 0 0 0]{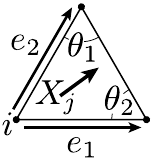},{}]
where \(A_f\) is the area of the face, \(N\) is its unit normal, \(e_i\) is the \(i\)th edge vector (oriented counter-clockwise), and \(u_i\) is the value of \(u\) at the opposing vertex.  The integrated divergence associated with vertex \(i\) can be written as
\end{window}
\[ \nabla \cdot X = \frac{1}{2} \sum_j \cot\theta_1 (e_1 \cdot X_j) + \cot\theta_2 (e_2 \cdot X_j) \]
where the sum is taken over incident triangles \(j\) each with a vector~\(X_j\), \(e_1\) and \(e_2\) are the two edge vectors of triangle \(j\) containing \(i\), and \(\theta_1, \theta_2\) are the opposing angles.  If we let \(d \in \mathbb{R}^{|V|}\) be the vector of (integrated) divergences of the normalized vector field \(X\), then the final distance function is computed by solving the symmetric Poisson problem
\[ L_C \phi = d. \]
Conveniently, this discretization easily generalizes to higher dimensions (\eg, tetrahedral meshes) using well-established discrete operators; see for instance~\cite{Desbrun:2006:DDF}.

\subsubsection{Polygonal Surfaces}
\label{sec:polygonal}

\begin{figure}[t]
\begin{center}
\includegraphics[width=\columnwidth]{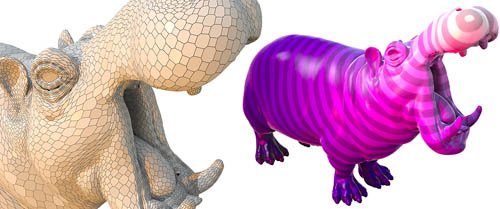}
\end{center}
\caption{Since the heat method is based on well-established discrete operators like the Laplacian, it is easy to adapt to a variety of geometric domains.  Above: distance on a hippo composed of high-degree nonplanar (and sometimes nonconvex) polygonal faces.}
\label{fig:hippo}
\end{figure}

For a mesh with (not necessarily planar) polygonal faces, we use the polygonal Laplacian defined by Alexa and Wardetzky~\shortcite{AW11}.  The only difference in this setting is that the gradient of the heat kernel is expressed as a discrete 1-form associated with half edges, hence we cannot directly evaluate the magnitude of the gradient \(|\nabla u|\) needed for the normalization step (\algo{heatmethod}, step II).  To resolve this issue we assume that \(\nabla u\) is constant over each face, implying that
\[ u_f^T L_f u_f = \int_M |\nabla u|^2 dA = |\nabla u|^2 A_f, \]
where \(u_f\) is the vector of heat values in face \(f\), \(A_f\) is the magnitude of the area vector, and \(L_f\) is the local (weak) Laplacian.  We can therefore approximate the magnitude of the gradient as
\[ |\nabla u|_f = \sqrt{u_f^T L_f u_f/A_f} \]
which is used to normalize the 1-form values in the corresponding face. The integrated divergence is given by \(\mathrm{d}^T M \alpha\) where \(\alpha\) is the normalized gradient, \(\mathrm{d}\) is the coboundary operator and \(M\) is the mass matrix for 1-forms (see \cite{AW11} for details).  These operators are applied in steps I-III as usual.  \figr{hippo} demonstrates distance computed on an irregular polygonal mesh.

\subsubsection{Point Clouds}
\label{sec:point}

\begin{figure}[t]
\begin{center}
\includegraphics[width=\columnwidth]{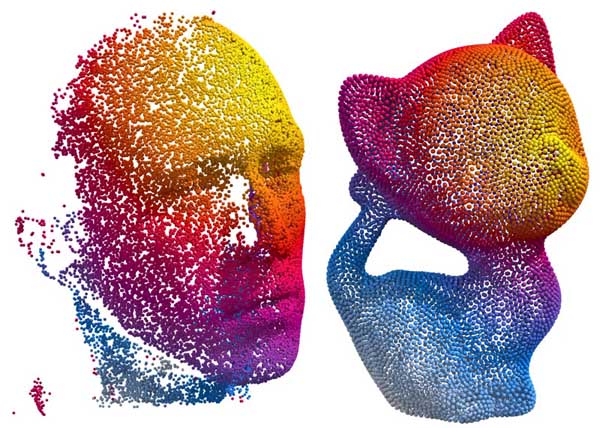}
\end{center}
\caption{The heat method can be applied directly to point clouds that lack connectivity information.  \emph{Left:} face scan with holes and noise.  \emph{Right:} kitten surface with connectivity removed.  Yellow points are close to the source; disconnected clusters (in the sense of Liu \etal) receive a constant red value.}
\label{fig:pointcloud}
\end{figure}

For a discrete collection of point samples \(P \subset \mathbb{R}^n\) of \(M\) with no connectivity information, we solve the heat equation (step I) using the symmetric \emph{point cloud} Laplacian recently introduced by Liu et al.~\shortcite{Liu11}, which extends previous work of Belkin et al.~\shortcite{Belkin09}. In this formulation, the Laplacian is represented by \(A_{\mathcal{V}}^{-1} L_{PC}\), where \(A_{\mathcal{V}}\) is a diagonal matrix of approximate Voronoi areas associated with each point, and \(L_{PC}\) is a symmetric positive semidefinite matrix (see~\cite{Liu11}, Section 3.4, for details).

To compute the vector field \(X = -\nabla u / |\nabla u|\) (step II), we represent the function \(u:P\to\mathbb{R}\) as a height function over approximate tangent planes \(T_p\) at each point \(p\in P\) and evaluate the gradient of a weighted least squares (WLS) approximation of \(u\)~\cite{Nealen:2003:AAS}. To compute tangent planes, we use a moving least squares (MLS) approximation for simplicity -- although other choices might be desirable (see Liu \etal).  The WLS approximation of \(\nabla u\) also provides a linear mapping \(u \mapsto D u\), taking any scalar function \(u\) to its gradient.  To find the best-fit scalar potential \(\phi\) (step III), we solve the linear, positive-semidefinite Poisson equation \( L_{PC} \phi = D^T A_{\mathcal{V}} X\).  The distance resulting from this approach is depicted in \figr{pointcloud}. 

Other discretizations are certainly possible (see for instance~\cite{Luo:2009:AGM}); we picked one that was simple to implement in any dimension. Note that the computational cost of the heat method depends primarily on the \emph{intrinsic} dimension~\(n\) of \(M\), whereas methods based on fast marching require a grid of the same dimension \(m\) as the ambient space~\cite{MemoliSapiro01} -- this distinction is especially important in contexts like machine learning where \(m\) may be significantly larger than \(n\).

\subsubsection{Choice of Time Step}

The accuracy of the heat method relies in part on the choice of time step \(t\).  In the smooth setting, \eqn{conv-Varadhan-2} suggests that smaller values of \(t\) yield better approximations of geodesic distance.  In the discrete setting we instead discover a rather remarkable fact, namely that the limit solution to \eqn{BVP-IE} is purely a function of the \emph{combinatorial} distance, independent of how we discretize the Laplacian (see \appn{graphvaradhan}).  The main implication of this fact is that -- \emph{on a fixed mesh} -- decreasing the value of \(t\) does not necessarily improve accuracy, even in exact arithmetic.  (Note, of course, that we can always improve accuracy by refining the mesh and decreasing \(t\) accordingly.)  Moreover, increasing the value of \(t\) past a certain point produces a smoothed approximation of geodesic distance (\sect{smoothed}).  We therefore seek an optimal time step \(t^*\) that is neither too large nor too small.

Determining a provably optimal expression for \(t^*\) is difficult due to the great complexity of analysis involving the cut locus~\cite{Neel:2004:ACL}.  We instead use a simple estimate that works remarkably well in practice, namely \(t = mh^2\) where \(h\) is the mean spacing between adjacent nodes and \(m > 0\) is a constant.  This estimate is motivated by the fact that \(h^2\Delta\) is invariant with respect to scale and refinement; experiments on a regular grid (\figr{smallT}) suggest that \(m=1\) is the smallest parameter value that recovers the \(\ell_2\) distance, and indeed this value yields near-optimal accuracy for a wide variety of irregularly triangulated surfaces, as demonstrated in \figr{optimalT}.  In this paper the time step
\[ \boxed{t=h^2} \]
is therefore used uniformly throughout all tests and examples, except where we explicitly seek a smoothed approximation of distance, as in \sect{smoothed}.

\subsection{Smoothed Distance}
\label{sec:smoothed}

\begin{figure}[b]
\begin{center}
\includegraphics[width=\columnwidth]{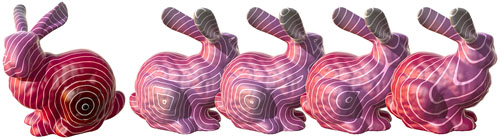}
\end{center}
\caption{A source on the front of the bunny results in nonsmooth cusps on the opposite side.  By running heat flow for progressively longer durations \(t\), we obtain smoothed approximations of geodesic distance \emph{(right)}.}
\label{fig:smoothing}
\end{figure}

\begin{figure}[t]
\begin{center}
\includegraphics[width=.8\columnwidth]{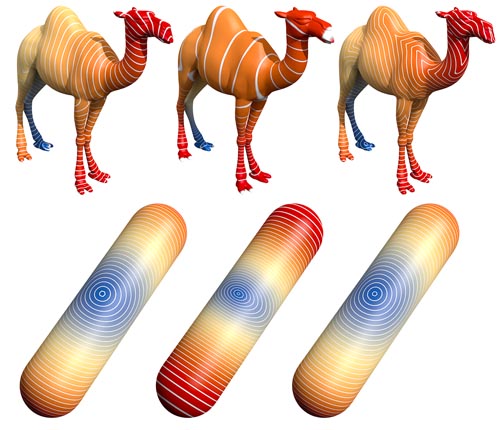}
\end{center}
\caption{\emph{Top row:} our smooth approximation of geodesic distance (left) and biharmonic distance (middle) both mitigate sharp ``cusps'' found in the exact distance (right), but notice that isoline spacing of the biharmonic distance can vary dramatically.  Bottom row: biharmonic distance (middle) tends to exhibit elliptical level lines near the source, while our smoothed distance (left) maintains isotropic circular profiles as seen in the exact distance (right).}
\label{fig:biharmonic}
\end{figure}

Geodesic distance fails to be smooth at points in the \emph{cut locus}, \ie, points at which there is no unique shortest path to the source -- these points appear as sharp cusps in the level lines of the distance function.  Non-smoothness can result in numerical difficulty for applications which need to take derivatives of the distance function \(\phi\) (\eg, level set methods), or may simply be undesirable aesthetically.

Several distances have been designed with smoothness in mind, including diffusion distance~\cite{Coifman06}, commute-time distance~\cite{Fouss:2007:RWC}, and biharmonic distance~\cite{Lipman:2010:BD} (see the last reference for a more detailed discussion).  These distances satisfy a number of important properties (smoothness, isometry-invariance, \etc), but are poor approximations of true geodesic distance, as indicated by uneven spacing of isolines (see \figr{biharmonic}, middle).  They can also be expensive to evaluate, requiring either a large number of Laplacian eigenvectors (\(\sim\!150-200\) in practice) or the solution to a linear system at each vertex.

In contrast, one can rapidly construct smoothed versions of geodesic distance by simply applying the heat method for large values of \(t\) (\figr{smoothing}).  The computational cost remains the same, and isolines are evenly spaced for any value of \(t\) due to normalization (step II).  Note that the resulting smooth distance function is isometrically (but not conformally) invariant since it depends only on the intrinsic Laplace--Beltrami operator.

Interestingly enough, existing smooth distance functions can also be understood in terms of time-discrete heat flow.  In particular, the commute-time distance \(d_C\) and biharmonic distance \(d_B\) can be expressed in terms of the harmonic and biharmonic Green's functions \(g_C\) and \(g_B\):
\[
   \begin{array}{rcl}
      d_C(x,y)^2 &=& g_C(x,x)-2g_C(x,y)+g_C(y,y),\\
      d_B(x,y)^2 &=& g_B(x,x)-2g_B(x,y)+g_B(y,y).
   \end{array}
\]
On a manifold of constant sectional curvature the sum \(g(x,x) + g(y,y)\) is constant, hence the commute-time and biharmonic distances are essentially a scalar multiple of the harmonic and biharmonic Green's functions (respectively), which can be expressed via one- and two-step backward Euler approximations of heat flow:
\[
   \begin{array}{rcl}
      g_C &=& \lim_{t \rightarrow \infty} (\id - t\Delta)^\dagger \delta,\\
      g_B &=& \lim_{t \rightarrow \infty} (\id - 2t\Delta + t^2\Delta^2)^\dagger \delta.
   \end{array}
\]
(Here {\small\(\dagger\)} denotes the pseudoinverse.)  Note that for finite \(t\) the identity operator acts as a regularizer, preventing a logarithmic singularity. For spaces with variable curvature, the Green's functions provide only an approximation of the corresponding distance functions.

\subsection{Boundary Conditions}
\label{sec:boundary}

If one is interested in the exact distance, either vanishing Neumann or Dirichlet conditions suffice since this choice does not affect the behavior of the smooth limit solution (see \cite{Renesse:2003:HKC}, Corollary 2 and \cite{Norris:1997:HKA}, Theorem 1.1, respectively).  Boundary conditions do however alter the behavior of our smoothed geodesic distance (\ie, large \(t\)) -- \figr{bcs} illustrates this behavior.  Although there is no well-defined ``correct'' behavior for the smoothed solution, we advocate the use of \emph{averaged} boundary conditions obtained as the mean of the Neumann solution \(u_N\) and the Dirichlet solution \(u_D\), \ie, \(u = \frac{1}{2}(u_N+u_D)\) -- these conditions tend to produce isolines that are not substantially influenced by the shape of the boundary.  The intuition behind this behavior is again based on interpreting heat diffusion in terms of random walks: zero Dirichlet conditions absorb heat, causing walkers to ``fall off'' the edge of the domain.  Neumann conditions prevent heat from flowing out of the domain, effectively ``reflecting'' random walkers.  Averaged boundary conditions mimic the behavior of a domain without boundary: the number of walkers leaving equals the number of walkers returning.  \figr{maze} shows how boundary conditions affect the behavior of geodesics in a path-planning scenario.    

\begin{figure}
\begin{center}
\includegraphics[width=.8\columnwidth]{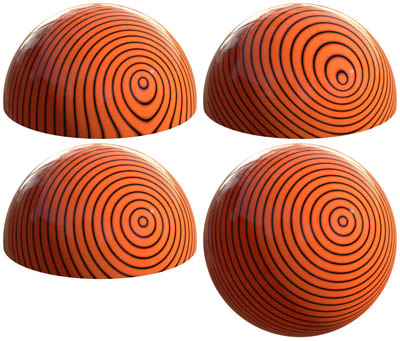}
\end{center}
\caption{Effect of Neumann (top-left), Dirichlet (top-right) and averaged (bottom-left) boundary conditions on smoothed distance.  Note that averaged conditions mimic the behavior of the same surface without boundary.}
\label{fig:bcs}
\end{figure}

\begin{figure}
\begin{center}
\includegraphics[width=.8\columnwidth]{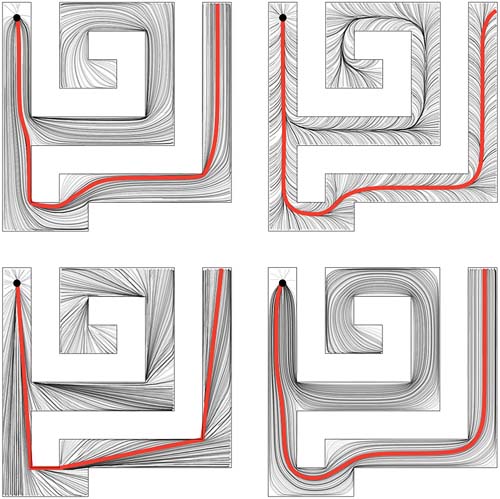}
\end{center}
\caption{For path planning, the behavior of geodesics can be controlled via boundary conditions and the integration time \(t\).  \emph{Top-left:} Neumann conditions encourage boundary adhesion.  \emph{Top-right:} Dirichlet conditions encourage avoidance.  \emph{Bottom-left:} small values of \(t\) yield standard straight-line geodesics.  \emph{Bottom-right:} large values of \(t\) yield more natural trajectories.}
\label{fig:maze}
\end{figure}

\section{Comparison}
\label{sec:comparison}

\begin{figure*}
\begin{center}
\includegraphics[width=\linewidth]{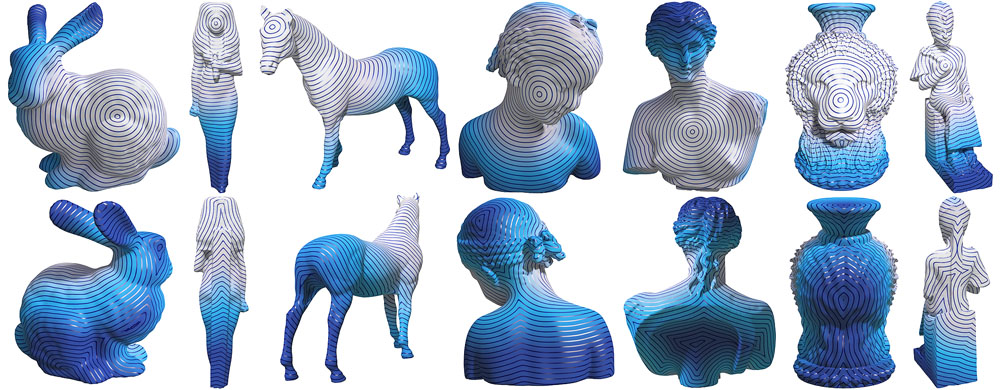}
\end{center}
\caption{Meshes used to test performance and accuracy (see \tabl{zoo}).  Left to right: \textsc{Bunny}, \textsc{Isis}, \textsc{Horse}, \textsc{Bimba}, \textsc{Aphrodite}, \textsc{Lion}, \textsc{Ramses}.}
\label{fig:zoo}
\end{figure*}

\begin{table*}[t]
\tbl{Comparison with fast marching and exact polyhedral distance.  Best speed/accuracy in \textbf{bold}; speedup in \textbf{\textcolor{orange}{orange}}.}{
\begin{tabular}{rc|cccc|ccc|c}
\hline
\textsc{Model} & \textsc{Triangles} & \multicolumn{4}{c}{\textsc{Heat Method}} & \multicolumn{3}{c}{\textsc{Fast Marching}} & \textsc{Exact} \\
\hline
& & \textsc{Precompute} & \textsc{Solve} & \textsc{Max Error} & \textsc{Mean Error} & \textsc{Time} & \textsc{Max Error} & \textsc{Mean Error} & \textsc{Time} \\
\hline
\textsc{Bunny}      & 28k  & \textbf{0.21s}  & \textbf{0.01s} \textbf{\textcolor{orange}{(28x)}}  & 3.22\%  &  \textbf{1.12\%}  &  0.28s  &  \textbf{1.06\%}  &   1.15\%          &   0.95s \\
\textsc{Isis}       & 93k  & \textbf{0.73s}  & \textbf{0.05s} \textbf{\textcolor{orange}{(21x)}}  & 1.19\%  &  \textbf{0.55\%}  &  1.06s  &  \textbf{0.60\%}  &   0.76\%          &   5.61s \\
\textsc{Horse}      & 96k  & \textbf{0.74s}  & \textbf{0.05s} \textbf{\textcolor{orange}{(20x)}}  & 1.18\%  &  \textbf{0.42\%}  &  1.00s  &  \textbf{0.74\%}  &   0.66\%          &   6.42s \\
\textsc{Kitten}     & 106k & \textbf{1.13s}  & \textbf{0.06s} \textbf{\textcolor{orange}{(22x)}}  & 0.78\%  &  \textbf{0.43\%}  &  1.29s  &  \textbf{0.47\%}  &   0.55\%          &  11.18s \\
\textsc{Bimba}      & 149k & \textbf{1.79s}  & \textbf{0.09s} \textbf{\textcolor{orange}{(29x)}}  & 1.92\%  &  0.73\%           &  2.62s  &  \textbf{0.63\%}  &   \textbf{0.69\%} &  13.55s \\
\textsc{Aphrodite}  & 205k & \textbf{2.66s}  & \textbf{0.12s} \textbf{\textcolor{orange}{(47x)}}  & 1.20\%  &  \textbf{0.46\%}  &  5.58s  &  \textbf{0.58\%}  &   0.59\%          &  25.74s \\
\textsc{Lion}       & 353k & \textbf{5.25s}  & \textbf{0.24s} \textbf{\textcolor{orange}{(24x)}}  & 1.92\%  &  0.84\%           & 10.92s  &  \textbf{0.68\%}  &   \textbf{0.67\%} &  22.33s \\
\textsc{Ramses}     & 1.6M & \textbf{63.4s}  & \textbf{1.45s} \textbf{\textcolor{orange}{(68x)}}  & 0.49\%  &  \textbf{0.24\%}  & 98.11s  &  \textbf{0.29\%}  &   0.35\%          & 268.87s \\
\hline
\end{tabular}
}
\label{tab:zoo}
\end{table*}

\subsection{Performance}
\label{sec:performance}

\begin{window}[0,r,\includegraphics[width=.45\columnwidth]{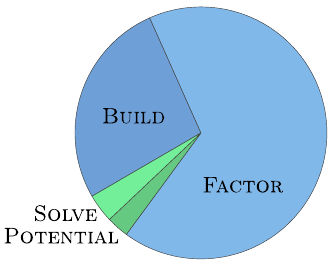},{}]
A key advantage of the heat method is that the linear systems in steps (I) and (III) can be prefactored.  Our implementation uses sparse Cholesky factorization~\cite{Chen:2008:ACS}, which for Poisson-type problems has guaranteed sub-quadratic complexity but in practice scales even better~\cite{Botsch05}; moreover there is strong evidence to suggest that sparse systems arising from elliptic PDEs can be solved in very close to linear time~\cite{Schmitz12,Spielman04}.  Independent of these issues, the amortized cost for problems with a large number of right-hand sides is roughly linear, since back substitution can be applied in essentially linear time.  See inset for a breakdown of relative costs in our implementation.
\end{window}

In terms of absolute performance, a number of factors affect the run time of the heat method including the spatial discretization, choice of discrete Laplacian, geometric data structures, and so forth.  As a typical example, we compared our simplicial implementation (\sect{simplicial}) to the first-order fast marching method of Kimmel \& Sethian~\shortcite{KimmelSethian98} and the exact algorithm of Mitchell \etal~\shortcite{Mitchell87} as described by Surazhsky \etal~\shortcite{Surazhsky05}.  In particular we used the state-of-the-art fast marching implementation of Peyr\'{e} and Cohen~\shortcite{Peyre:2005:GCF} and the exact implementation of Kirsanov~\cite{Surazhsky05}.  The heat method was implemented in ANSI C using a simple vertex-face adjacency list.  All timings were taken on a 2.4 GHz Intel Core 2 Duo machine using a single core -- \tabl{zoo} gives timing information.  Note that even for a single distance computation the heat method outperforms fast marching; more importantly, updating distance via the heat method for new subsets \(\gamma\) is consistently an order of magnitude faster (or more) than both fast marching and the exact algorithm.

\subsection{Accuracy}

\begin{figure}
\begin{center}
\includegraphics[width=.9\columnwidth]{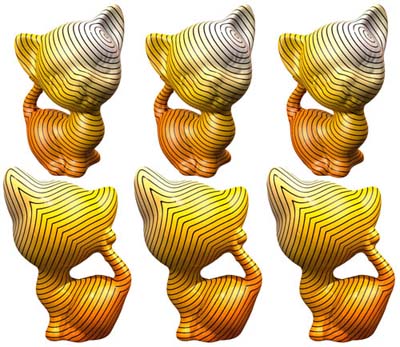}
\end{center}
\caption{Visual comparison of accuracy.  \emph{Left:} exact geodesic distance.  Using default parameters, the heat method \emph{(middle)} and fast marching (right) both produce results of comparable accuracy, here within less than 1\% of the exact distance -- see \tabl{zoo} for a more detailed comparison.}
\label{fig:compare}
\end{figure}

We examined errors in the heat method, fast marching~\cite{KimmelSethian98}, and the exact polyhedral distance~\cite{Mitchell87}, relative to mean edge length \(h\) for a variety of triangulated surfaces.  \figrs{sphere_convergence}{torus_convergence} illustrate the rate of convergence on simple geometries where the smooth geodesic distance can be easily obtained.  Both fast marching and the heat method appear to exhibit linear convergence; it is also interesting to note that the exact \emph{polyhedral} distance provides only quadratic convergence.  Keeping this fact in mind, \tabl{zoo} uses the polyhedral distance as a baseline for comparison on more complicated geometries -- here \textsc{Max} is the maximum error as a percentage of mesh diameter and \textsc{Min} is the mean relative error at each vertex (a convention introduced in~\cite{Surazhsky05}).  Note that fast marching tends to achieve a smaller maximum error, whereas the heat method does better on average.  \figr{compare} gives a visual comparison of accuracy; the only notable discrepancy is a slight smoothing at sharp cusps; \figr{hiragana_a} indicates that this phenomenon does not interfere with the extraction of the cut locus -- here we simply threshold the magnitude of \(\Delta\phi\).  \figr{metric_convergence} plots the maximum violation of metric properties -- both the heat method and fast marching exhibit small approximation errors that vanish under refinement.  Even for the smoothed distance (\(m >> 1\)) the triangle inequality is violated only for highly degenerate geodesic triangles, \ie, all three vertices on a common geodesic.  (In contrast, smoothed distances discussed in \sect{related} satisfy metric properties exactly, but cannot be used to obtain the true geometric distance.)  Overall, the heat method exhibits errors of the same magnitude and rate of convergence as fast marching (at lower computational cost) and is likely suitable for any application where fast marching is presently used.

The accuracy of the heat method might be further improved by considering alternative spatial discretizations (see for instance~\cite{Belkin:2009:DLO,Hildebrandt:2011:LBO}), though again one should note that even the exact polyhedral distance yields only an \(O(h^2)\) approximation.  In the case of the fast marching method, accuracy is determined by the choice of \emph{update rule}.  A number of highly accurate update rules have been developed in the case of regular grids (\eg, HJ WENO~\cite{Jiang:1997:WES}), but fewer options are available on irregular domains such as triangle meshes, the predominant choice being the first-order update rule of Kimmel and Sethian~\shortcite{KimmelSethian98}.  Finally, the approximate algorithm of Surazhsky \etal\ provides an interesting comparison since it is on par with fast marching in terms of performance and produces more accurate results (see \cite{Surazhsky05}, Table 1).  Similar to fast marching, however, it does not take advantage of precomputation and therefore exhibits a significantly higher amortized cost than the heat method; it is also limited to triangle meshes.

\begin{figure}[ht]
\begin{center}
\includegraphics[width=\columnwidth]{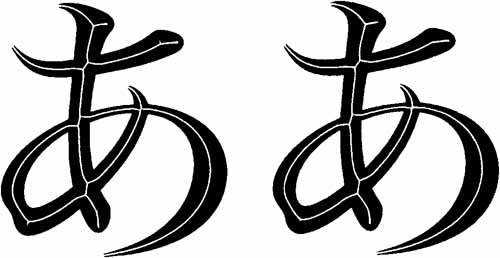}
\end{center}
\caption{Medial axis of the hiragana letter ``a'' extracted by thresholding second derivatives of the distance to the boundary.  Left: fast marching.  Right: heat method.\label{fig:hiragana_a}}
\end{figure}

\subsection{Robustness}

Two factors contribute to the robustness of the heat method, namely (1) the use of an unconditionally stable implicit time-integration scheme and (2) formulation in terms of purely elliptic PDEs.  \figr{robustness} verifies that the heat method continues to work well even on meshes that are poorly discretized or corrupted by a large amount of noise (here modeled as uniform Gaussian noise applied to the vertex coordinates).  In this case we use a moderately large value of \(t\) to investigate the behavior of our smoothed distance; similar behavior is observed for small \(t\) values. \figr{femme} illustrates the robustness of the method on a surface with many small holes as well as long sliver triangles.

\begin{figure}[b]
\begin{center}
\includegraphics[width=.9\columnwidth]{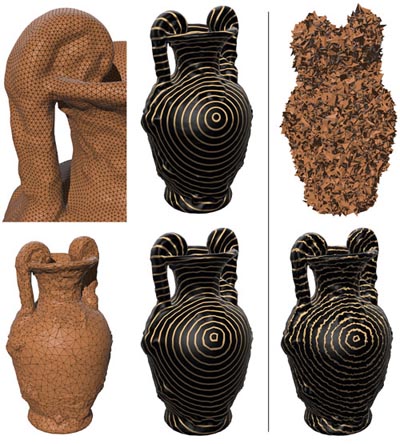}
\end{center}
\caption{Tests of robustness.  \emph{Left:} our smoothed distance (\sect{smoothed}) appears similar on meshes of different resolution. \emph{Right:} even for meshes with severe noise \emph{(top)} we recover a good approximation of the distance function on the original surface \emph{(bottom, visualized on noise-free mesh)}.}
\label{fig:robustness}
\end{figure}

\begin{figure}[ht!]
\begin{center}
\includegraphics[width=\columnwidth]{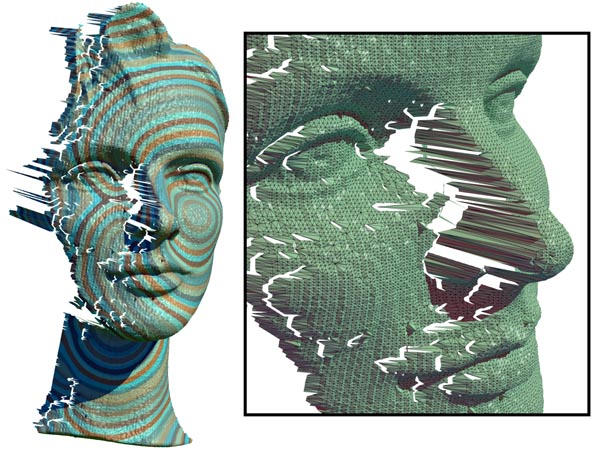}
\end{center}
\caption{Smoothed geodesic distance on an extremely poor triangulation with significant noise -- note that small holes are essentially ignored.  Also note good approximation of distance even along thin slivers in the nose.\label{fig:femme}}
\end{figure}

\section{Conclusion}

The heat method is a simple, general method that can be easily incorporated into a broad class of algorithms.  However, a great deal remains to be explored, including an investigation of alternative discretizations.  Further improvements on the optimal \(t\) value also provide an interesting avenue for future work, though typically the existing estimate already outperforms fast marching in terms of mean error (\tabl{zoo}).  Another obvious question is whether a similar transformation can be applied to a larger class of Hamilton-Jacobi equations.  Finally, \emph{weighted} distance computation might be achieved by simply rescaling the source data.

\appendix

\section{A Varadhan Formula for Graphs}
\label{app:graphvaradhan}

\begin{lemma}
Let \(G=(V,E)\) be the graph induced by the sparsity pattern of any real symmetric matrix \(A\), and consider the linear system
\[ (I-tA)u_t = \delta \]
where \(I\) is the identity, \(\delta\) is a Kronecker delta at a source vertex \(u \in V\), and \(t>0\) is a real parameter.  Then generically
\[ \boxed{\phi = \lim_{t \rightarrow 0} \frac{\log u_t}{\log t}} \]
where \(\phi \in \mathbb{N}_0^{|V|}\) is the \textbf{graph distance} (\ie, number of edges) between each vertex \(v \in V\) and the source vertex \(u\).
\label{lem:smallT}
\end{lemma}

\begin{proof}
Let \(\sigma\) be the operator norm of \(A\).  Then for \(t < 1/\sigma\) the matrix \(B:=I-tA\) has an inverse and the solution \(u_t\) is given by the convergent Neumann series \(\sum_{k=0}^\infty t^k A^k \delta.\)  Let \(v \in V\) be a vertex \(n\) edges away from \(u\), and consider the ratio \(r_t := |s|/|s_0|\) where \(s_0 := (t^n A^n \delta)_v\) is the first nonzero term in the sum and \(s = (\sum_{k=n+1}^\infty t^k A^k \delta)_v\) is the sum of all remaining terms.  Noting that \(|s| \leq \sum_{k=n+1}^\infty t^k ||A^k \delta|| \leq \sum_{k=n+1}^\infty t^k \sigma^k,\) we get
\[r_t \leq \frac{t^{n+1}\sigma^{n+1}\sum_{k=0}^\infty t^k \sigma^k}{t^n(A^n \delta)_v} = c\frac{t}{1-t\sigma},\]
where the constant \(c := \sigma^{n+1}/(A^n \delta)_v\) does not depend on \(t\).  We therefore have \(\lim_{t \rightarrow 0} r_t = 0\), \ie, only the first term \(s_0\) is signifcant as \(t\) goes to zero.  But \(\log s_0 = n \log t + \log (A^n \delta)_v\) is dominated by the first term as \(t\) goes to zero, hence \(\log (u_t)_v/\log t\) approaches the number of edges \(n\).
\end{proof}

Numerical experiments such as those depicted in \figr{smallT} agree with this analysis.

\begin{figure}[h]
\begin{center}
\includegraphics[width=\columnwidth]{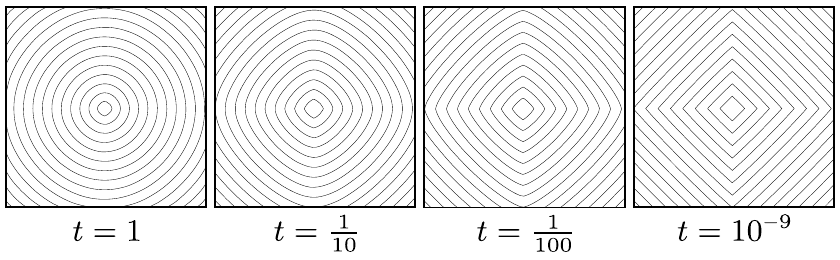}
\caption{Isolines of \(\log u_t/\log t\) computed in exact arithmetic on a regular grid with unit spacing (\(h=1\)).  As predicted by \lemm{smallT}, the solution approaches the combinatorial distance as \(t\) goes to zero.\label{fig:smallT}}
\end{center}
\end{figure}

\begin{figure}
\begin{center}
\includegraphics[width=\columnwidth]{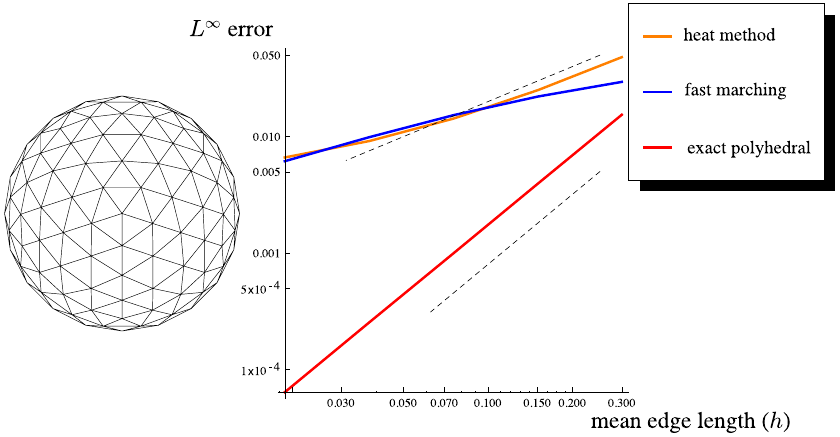}
\end{center}
\caption{\(L^\infty\) convergence of distance functions on the unit sphere with respect to mean edge length.  As a baseline for comparison, we use the exact distance function \(\phi(x,y) = \cos^{-1}(x \cdot y)\).  Linear and quadratic convergence are plotted as dashed lines for reference; note that even the exact polyhedral distance converges only quadratically.\label{fig:sphere_convergence}}
\end{figure}

\begin{figure}
\begin{center}
\includegraphics[width=\columnwidth]{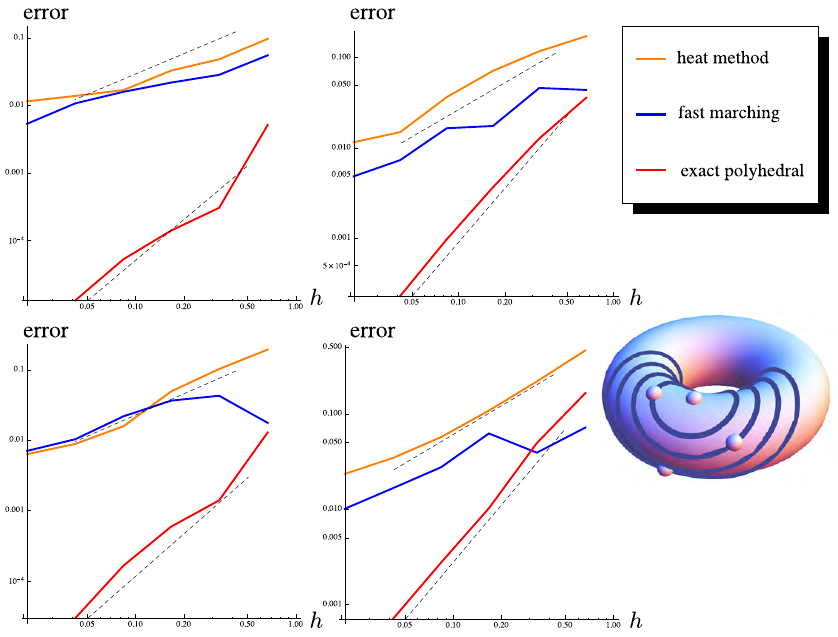}
\caption{Convergence of geodesic distance on the torus at four different test points.  Error is the absolute value of the difference between the numerical value and the exact (smooth) distance; linear and quadratic convergence are plotted as dashed lines for reference. \emph{Right:} test points visualized on the torus; dark blue lines are geodesic circles computed via Clairaut's relation.\label{fig:torus_convergence}}
\end{center}
\end{figure}

\begin{figure}
\begin{center}
\includegraphics[width=\columnwidth]{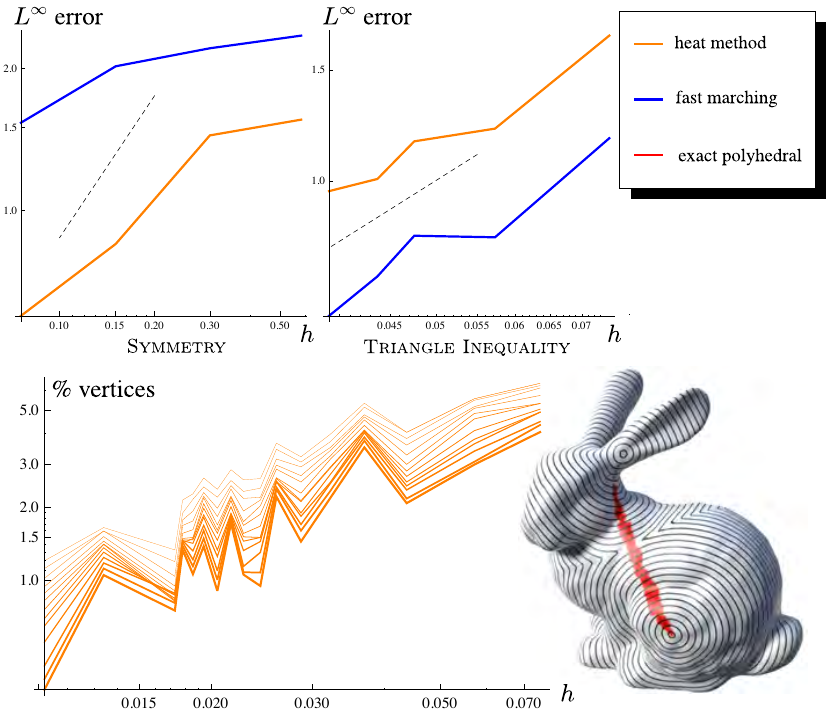}
\caption{Numerical approximations of geodesic distance exhibit small violations of metric properties that vanish under refinement.  Here we examine errors in symmetry \emph{(top left)} and the triangle inequality \emph{(top right)} by checking all pairs or triples (respectively) of vertices on the Stanford bunny and plotting the worst violation as a percent of mesh diameter.  Linear convergence is plotted as a dashed line for reference. \emph{Bottom right:} violation of triangle inequality occurs only for degenerate geodesic triangles, \ie, all three vertices along a common geodesic.  Fixing the first two vertices, we plot those in violation in \textbf{\color{red}{red}}.  \emph{Bottom left:} percent of vertices in violation; letting \(t=mh^2\), each curve corresponds to a value of \(m\) sampled from the range \([1,100]\).\label{fig:metric_convergence}}
\end{center}
\end{figure}

\begin{figure}
\begin{center}
\includegraphics[width=\columnwidth]{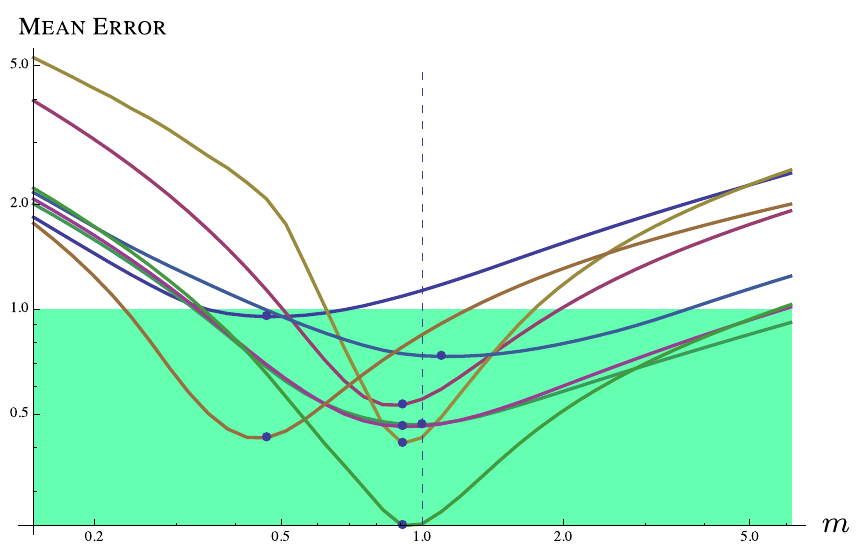}
\end{center}
\caption{Mean percent error as a function of \(m\), where \(t = mh^2\).  Each curve corresponds to a data set from \tabl{zoo}.  Notice that in most examples \(m=1\) (dashed line) is close to the optimal parameter value (blue dots) and yields mean error below \(1\%\).\label{fig:optimalT}}
\end{figure}

\begin{acks}
To be included in final version.
\end{acks}

\received{September 2012}{TBD}

\end{document}